\documentclass[12pt]{amsart}

\usepackage{newunicodechar}
\usepackage{quiverSDG}
\usepackage{fontawesome}
\usepackage{soul}
\usepackage{cancel}
\usepackage{xspace}
\usepackage{eucal}
\usepackage{colonequals}
\usepackage{graphicx}
\usepackage{afterpage}
\usepackage[pdfauthor={Riccardo Brasca, Gabriella Clemente},pdftitle={Synthetic Differential Geometry in Lean},linktocpage, colorlinks=true, allcolors=blue, pdftex, pdfpagelabels=false, bookmarks, hyperindex, hyperfigures]{hyperref}
\hypersetup{pdfkeywords={Synthetic differential geometry, Differential geometry, Formalization of Mathematics, mathlib, lean},
  pdfsubject={MSC 2020: 18F40, 53B99, 68V20, 68V15}}
\definecolor{darkergreen}{rgb}{0.0, 0.5, 0.0}
\definecolor{Blue}{RGB}{0,0,148}

\usepackage{listings}

\definecolor{light-gray}{gray}{0.95}

\usepackage{tikz}
\usetikzlibrary{positioning, arrows.meta, shapes.geometric}

\definecolor{keywordcolor}{rgb}{0.7, 0.1, 0.1}   
\definecolor{commentcolor}{rgb}{0.4, 0.4, 0.4}   
\definecolor{symbolcolor}{rgb}{0, 0, 0.8}    
\definecolor{tacticcolor}{rgb}{0, 0, 0.8}    
\definecolor{sortcolor}{rgb}{0.1, 0.5, 0.1}      
\newcommand*{\lean}[1]{\lstinline{#1}\xspace} 
\theoremstyle{plain}
\newtheorem{theorem}{Theorem}[subsection]

\newtheorem{corollary}[theorem]{Corollary}
\newtheorem*{principle}{Principle}
\newtheorem{lemma}[theorem]{Lemma}
\newtheorem{proposition}[theorem]{Proposition}

\theoremstyle{definition}
\newtheorem{definition}[theorem]{Definition}
\newtheorem{remark}[theorem]{Remark}

\newcommand*{\N}{\mathbb{N}}
\newcommand*{\Z}{\mathbb{Z}}
\newcommand*{\Q}{\mathbb{Q}}

\DeclareMathOperator{\id}{id}

\DeclareMathOperator{\Hom}{Hom}
\DeclareMathOperator{\Alg}{Alg}
\DeclareMathOperator{\ev}{ev}

\DeclareMathOperator{\man}{\textbf{Man}}

\newcommand*{\mathlib}{mathlib\xspace}
\newcommand*{\Lean}{\textsc{Lean}\xspace} 

\makeatletter
\newcommand\leanlink{\begingroup\catcode`\#=12\relax\@leanlink}
\newcommand\@leanlink[2]{\endgroup
\href{#1}
{\texttt{\detokenize{#2}}}}

\title{Synthetic Differential Geometry in Lean}

\author{Riccardo Brasca}
\email{\href{mailto:riccardo.brasca@imj-prg.fr}{riccardo.brasca@imj-prg.fr}}
\urladdr{\url{http://www.imj-prg.fr/~riccardo.brasca/}}
\address{Université Paris Cité and Sorbonne Université, CNRS, IMJ-PRG, F-75013 Paris, France}
\author{Gabriella Clemente}
\email{\href{mailto:gabriella.clemente@cnrs.fr}{gabriella.clemente@cnrs.fr}}
\urladdr{\url{https://sites.google.com/view/gclemente}}
\address{Université Paris Cité, CNRS, Institut de recherche en informatique fondamentale, Paris, France}

\date{\today}



\begin{document}
\begin{abstract}
This article is about the formalization of synthetic differential geometry with the \Lean proof assistant and the mathematical library \mathlib. The main result we prove and formalize is a Taylor theorem for functions of several variables, where the series expansion is around an infinitesimal neighborhood. Most of our proofs are in fact new. Our investigations highlight the possibility of using \mathlib to do constructive mathematics.
\end{abstract}

\maketitle

\section*{Introduction}
Differential geometry (DG) exists in many forms. Our focus is synthetic differential geometry (SDG), which provides an axiomatic treatment of DG. The theme of our investigations is the formalization of SDG in \Lean, specifically the synthetic differential calculus from \cite{SDG}, which is covered in Part I of that book. The main results we prove are the synthetic Taylor formulas in one and in several variables (Theorems \ref{taylor_k} and \ref{taylor_multi}). These statements are not new: they appear in \cite[Part I]{SDG}, and complete proofs were given by Kock in \cite{Kock1978}, in the more general setting of ring objects of line type. What is new are the proofs and their formalization. Our proofs are elementary and use only ring computations in $R$ and the Kock--Lawvere conditions, so that every step can be carried out constructively in \Lean. Our work can also be viewed as an experiment aimed at understanding whether \Lean and its mathematical library \mathlib can be effectively used in a constructive setting. In principle, \Lean is not inherently a classical system: the axiom of choice is an explicit axiom of its standard library, and all classical reasoning passes through it (see Section~\ref{implementation}). However, \mathlib has largely been written without any systematic attempt to avoid classical reasoning. One of the goals of our project is therefore to investigate whether nontrivial constructive mathematics can be carried out using the existing infrastructure of \mathlib, without requiring large portions of the library to be rewritten. Our aim is not to provide a definitive answer to this question. Rather, we show through concrete examples that it is indeed possible to work with \Lean and \mathlib in a more constructive style. Achieving this, however, requires some additional effort, and the amount of work involved is likely to grow significantly when relying on more advanced parts of the library.

The project is available on \href{https://github.com/riccardobrasca/SDG}{GitHub}. Most mathematical statements and definitions will be accompanied by a direct link to the source code for the corresponding statement in \mathlib or in our repository. To keep the links usable, they all point to a fixed commit (the most recent one at the time of writing).

The mathematical motivation of our work is to understand the applications of formalized SDG for studying classical DG. In principle, formalized SDG could serve as a conjecture-making and conjecture-proving aid in local DG. Compared to classical DG, SDG is computationally less involved, and this can simplify the process of conjecturing synthetic preliminary versions of DG theorems. Formalization can mechanize the generalization and unification of results, which can be particularly useful when technicalities are too complex to work out by hand. Although we do not address these potential applications of formalized SDG, we believe that the topic deserves careful consideration.

Traditionally, DG takes place in the category of smooth manifolds and smooth maps between them. In contrast, SDG takes place in a smooth topos \cite{smooth_topos}, which is a special kind of Grothendieck topos; such topoi are constructed in \cite[Part III]{SDG} and \cite{MoerdijkReyes}. We discuss the definition very briefly as we did not formalize it. The definition is an external categorical one, and we explain below exactly what we assume. Let $R$ be a commutative ring with unity. An $R$-algebra $W$ is a Weil algebra if it is free of finite rank as an $R$-module and has a nilpotent ideal $J \subseteq W$ such that $W \simeq R \oplus J$ as $R$-modules (see \cite[I.16]{SDG} and \cite{smooth_topos}); such a $W$ is automatically finitely presented as an $R$-algebra. For instance, the $R$-algebra ${R[x]}/{\langle x^{k+1} \rangle}$ is a Weil algebra since
\[
{R[x]}/{\langle x^{k+1} \rangle} \simeq R \oplus (Rx \oplus Rx^2 \oplus \dots \oplus R x^k)
\]

A Grothendieck topos $\mathcal{E}$ is smooth if

\begin{enumerate}
\item it has a distinguished commutative unital ring object $R$,

\item for each Weil algebra $W$ the functor $(\cdot)^{\Hom_{R-\Alg}(W,R)}:\mathcal{E} \to \mathcal{E}$ has a right adjoint, and

\item the following axiom schema, known as the Kock--Lawvere axiom schema, holds true:

for any Weil algebra $W$, the evaluation map
\begin{gather*}
\ev \colon W \to R^{\Hom_{R-\Alg}(W,R)} \\
w \mapsto (f \mapsto \ev_w (f) \colonequals f(w))
\end{gather*}
is an isomorphism of $R$-algebras.
\end{enumerate}
In (2) and (3), $W$ varies across Weil algebras over the ring object $R$, viewed as $R$-algebra objects of $\mathcal{E}$. This is the convention of \cite{smooth_topos}. For instance, $W$ can be of the form $R \otimes W_0$, for $W_0$ a Weil algebra over $\Z$ (or over a base ring $k$ in which case $R$ is a $k$-algebra). These are the Weil algebras considered in \cite[I.16]{SDG}. Since $W$ is finitely presented, say $W = {R[x_1, \dots, x_m]}/{\langle f_1, \dots, f_r \rangle}$, the object $\Hom_{R-\Alg}(W,R)$ of $R$-algebra homomorphisms is the subobject $\{r \in R^m \mid f_1(r) = \dots = f_r(r) = 0\}$ of $R^m$, and (3) is an axiom schema in the internal language of $\mathcal{E}$, with one instance for each $W$. Condition (2), on the other hand, is a genuinely external property of $\mathcal{E}$, which does not translate into an axiom of the internal language (see \cite{Riley} and the references therein). So condition (2) plays no role in this paper.

SDG can be seen as the study of smooth topoi. Sections~\ref{basics}, \ref{1var} and \ref{multi-var} are written in the naive set-theoretic style of \cite[Part I]{SDG}. They are to be read in the internal intuitionistic logic of a topos $\mathcal{E}$ as above or, as we explain below, of any topos with a ring object satisfying the axioms we actually use. In particular, quantifiers over elements, functions and natural numbers are internal. The categorical background of this introduction and the comparison with classical calculus in Section~\ref{SvsO}, on the other hand, are external. The Weil algebras $W$ appearing in (2) and (3) are $R$-algebra objects of $\mathcal{E}$, but the quantification over them is external: (2) and (3) are schemas, with one instance for each $W$, rather than single internal statements. To each $W$, one associates an object of $\mathcal{E}$, the infinitesimal object $\Hom_{R-\Alg}(W,R)$. As explained above, for $W = {R[x_1, \dots, x_m]}/{\langle f_1, \dots, f_r \rangle}$, this is the subobject of $R^m$ on which $f_1, \dots, f_r$ vanish, that is, the equalizer of the two maps $R^m \to R^r$ given by $(f_1, \dots, f_r)$ and by $0$. It is thus a finite limit of finite powers of $R$, and it exists in any topos equipped with a ring object. Similarly, since $\mathcal{E}$ is Cartesian closed, the exponential $R^{\Hom_{R-\Alg}(W,R)}$ appearing in (3) is an object of $\mathcal{E}$, and so is $X^{\Hom_{R-\Alg}(W,R)}$ for any object $X$ of $\mathcal{E}$.

The only Weil algebras we deal with here are ${R[x]}/{\langle x^{k+1} \rangle}$ for $k \in \N$. An $R$-algebra homomorphism ${R[x]}/{\langle x^{k+1} \rangle} \to R$ is determined by the image of $x$, which is an element $r \in R$ with $r^{k+1} = 0$, so the corresponding infinitesimal object is
\[
\Hom_{R-\Alg} \big({R[x]}/{\langle x^{k+1} \rangle},R \big) \simeq \{r \in R \mid r^{k+1}=0\}.
\]
The right hand side is the description we actually work with. In Definition \ref{SDG.D} we set $D_k = \{r \in R \mid r^{k+1}=0\}$ for an arbitrary commutative ring $R$. This is a definition about the ring $R$ alone: it involves neither Weil algebras nor a topos, and it is what we formalize in \Lean. When $R$ is a ring object of a topos $\mathcal{E}$ and the definition is read internally, it recovers the infinitesimal object above. The most important case is the first nontrivial one, $D\colonequals D_1$, which corresponds to the dual numbers ${R[x]}/{\langle x^2 \rangle}$. In classical terms, $D$ should be thought of as an infinitely small interval $(-\epsilon,\epsilon) \subseteq \mathbb{R}$. In a smooth topos, $D$ is used to define the tangent bundle of any object $X$ as the exponential $T_X\colonequals X^D$, which should therefore be thought of as the space of all infinitesimal paths in $X$.

To illustrate how this internal reading works, consider the Kock--Lawvere axiom for $W = {R[x]}/{\langle x^2 \rangle}$. Identifying $W$ with $R \times R$ via $a + bx \mapsto (a, b)$, it says that the morphism
\[
\alpha \colon R \times R \to R^D, \qquad (a, b) \mapsto (d \mapsto a + bd)
\]
is an isomorphism (see Remark \ref{alpha} for the equivalence between this formulation and the elementwise one of Definition \ref{IsKockLawvere}). Applying $\Hom_{\mathcal{E}}(T, -)$ for an arbitrary object $T$ of $\mathcal{E}$ and using the exponential adjunction, we see that this is equivalent to the following statement: for every $g \colon T \times D \to R$ there are unique $a, b \colon T \to R$ such that
\[
g(t, d) = a(t) + b(t) d
\]
for all $t \in T$ and $d \in D$. In other words, the elementwise formulation of the axiom in Definition \ref{IsKockLawvere} (``every $g \colon D \to R$ is of the form $d \mapsto g(0) + bd$ for a unique $b \in R$'') has to be understood for $g$ depending on arbitrary parameters $t \in T$, and not only for global elements. This is what allows us to apply it to $g(x, d) = f(x + d)$, where $x \in R$ is a variable, to define the derivative of a function $f \colon R \to R$ (see Definition \ref{deriv}). We refer to \cite[Proposition II.5.1]{SDG} for details.

In practice, we never work in a topos, and we use neither condition (2) nor the Kock--Lawvere axiom for a general $W$. Besides nontriviality and, where needed, invertibility of positive integers, our only assumptions on $R$ are the two conditions of Definition \ref{IsKockLawvere}. These are the $1$-Kock--Lawvere condition, which is the instance $W = {R[x]}/{\langle x^2 \rangle}$ of (3) and is Axiom $1$ of \cite[Part I]{SDG}, and the Kock--Lawvere condition, which is the instance $W = {R[x]}/{\langle x^{k+1} \rangle}$ of (3) for all $k$ and is Axiom $1'$ of \cite[Part I]{SDG}. Each section of the paper states which of the two it assumes. In particular, our results hold for any ring object satisfying Axiom $1'$, for example in every smooth topos and in the well adapted models of \cite[Part III]{SDG} and \cite{MoerdijkReyes}. We also use the axiom of unique choice, which is valid in the internal logic of every topos \cite[Proposition II.3.4]{SDG}, but which has to be added explicitly to \Lean as explained in Section~\ref{implementation}. Formally, the results of Sections~\ref{basics}, \ref{1var} and \ref{multi-var} are theorems of \Lean's type theory without the axiom of choice and with the axiom of unique choice, about a ring $R$ satisfying Definition \ref{IsKockLawvere}.

Smooth topoi are related to the category $\man$ of smooth manifolds via well adapted models. A smooth topos is called a well adapted model if it is equipped with a fully faithful functor $\man \hookrightarrow \mathcal{E}$ with the property that $\mathbb{R} \mapsto R$, satisfying certain additional compatibility conditions (see \cite[Part III, Sections 3--4]{SDG}). We plan to investigate smooth topoi, well adapted models and their formalization in a forthcoming work.

The organization of the article is as follows. In Section~\ref{SvsO} we compare the synthetic and classical differential calculus theories. In Section~\ref{implementation}, we explain the fundamentals of formalizing SDG in \Lean. In Section~\ref{basics} we introduce all of the concepts that will be needed later on. In Section~\ref{1var} we develop the synthetic differential calculus theory in $1$ variable. Finally, in Section~\ref{multi-var} we treat the several variables theory. In the rest of the article, all indexing starts at $0$: if $r\in R^n$, then $r = (r_0,\dots,r_{n-1})$. This convention is in line with our \Lean code.

\subsection*{Acknowledgments}
We are grateful to the anonymous referees for their careful reading of the paper and for their many valuable comments, which greatly improved the exposition. We acknowledge Robin Arnez, Paul Lezeau and Damiano Testa for helping us in writing various scripts that made this work easier to complete. We would also like to thank the whole \Lean community, especially Anatole Dedecker, Jovan Gerbscheid, Bhavik Mehta, Filippo Nuccio and Shuhao Song, for various useful discussions. Riccardo Brasca is supported by the ANR project FALSE (ANR-25-CE40-7639). Gabriella Clemente is supported by the European Research Council (ERC) under the European Union's Ninth Framework Programme Horizon Europe (ERC Synergy Project Malinca, Grant Agreement n. 101167526).

\section{Synthetic versus ordinary differential calculus}\label{SvsO}
\setcounter{subsection}{1}

The main distinction between synthetic and ordinary differential calculus can be summarized as follows: in synthetic differential calculus, functions coincide with their Taylor polynomials on infinitesimal neighborhoods, whereas in ordinary differential calculus Taylor polynomials only provide approximations to functions. The infinitesimal neighborhoods of interest to us involve the $D_k$s as well as their pointwise sums and Cartesian products. The equality between Taylor polynomials and functions is a consequence of the Kock--Lawvere axiom schema. The instance we use is captured by Definition \ref{IsKockLawvere}. Intuitively, this Kock--Lawvere axiom says that all maps $D_k \to R$ are polynomials of degree at most $k$. In particular, on $D$, all maps are affine linear.

Let $U \subseteq \mathbb{R}^n$ be an open subset, $f \colon U \to \mathbb{R}$ be a function, and $k\geq 1$ be an integer. Recall that $f$ is said to be of class $C^k$, and we write $f\in C^k(U,\mathbb{R})$, if the first $k$ partial derivatives of $f$ exist and are continuous: for every multi-index $\alpha=(\alpha_0,\dots,\alpha_{n-1}) \in \N^n$ such that $\alpha_0+\dots+\alpha_{n-1}\leq k$,
\[
\frac{\partial^{\alpha_0+\dots+\alpha_{n-1}} f}{\partial x^{\alpha_0}_0 \dots \partial x^{\alpha_{n-1}}_{n-1}}
\]
exists and is continuous on $U$. We say that $f$ is smooth, and write $f\in C^{\infty}(U,\mathbb{R})$, if for all integers $k\geq 1$, $f\in C^k(U,\mathbb{R})$. Note that if $l \leq k$, then $C^k(U,\mathbb{R}) \subseteq C^l(U,\mathbb{R})$.

\begin{theorem}[Taylor's theorem in ordinary differential calculus]\label{TaylorThm}
Let $f\in C^{k+1}(U,\mathbb{R})$. For any $l\leq k$ and any $x_0\in U$,

\begin{gather*}
f(x_0+h)=\\
\sum_{0\leq \alpha_0+\dots+\alpha_{n-1} \leq l} \frac{1}{\alpha_0! \dots \alpha_{n-1}!}\frac{\partial^{\alpha_0+\dots+\alpha_{n-1}} f}{\partial x^{\alpha_0}_0 \dots \partial x^{\alpha_{n-1}}_{n-1}}(x_0) h^{\alpha_0}_0 \dots h^{\alpha_{n-1}}_{n-1}+ O(\|h\|^{l+1})
\end{gather*}
as $\|h\| \to 0$
\end{theorem}

The natural number $l$ is called the order of the Taylor polynomial approximation. Let us now compare the classical picture with the synthetic one. We prove and formalize Theorem \ref{taylor_multi}, among several other results. Theorem \ref{taylor_multi} is a synthetic counterpart of Taylor's theorem for multivariate functions (Theorem \ref{TaylorThm}).

Theorem \ref{TaylorThm} approximates $C^{k+1}$ functions by Taylor polynomials of degree at most $k$, whereas the synthetic theorems give exact identities on infinitesimal neighborhoods. To make this distinction explicit, suppose that $R$ is a $1$-Kock--Lawvere ring (Definition \ref{IsKockLawvere}). Theorem \ref{taylor_one} states that for every function $f\colon R\to R$, every $x\in R$ and every $d\in D$,
\[
f(x+d)=f(x)+\partial f(x)d,
\]
where $\partial f$ denotes the synthetic derivative (Definition \ref{deriv}). In ordinary calculus, if $g\in C^2(U,\mathbb{R})$ for an open subset $U\subseteq\mathbb{R}$ and $x\in U$ is fixed, the case $n=l=1$ of Theorem \ref{TaylorThm} gives
\[
g(x+h)=g(x)+g'(x)h+O(|h|^2)\qquad\text{as }h\to0.
\]
Thus the synthetic formula is an equality for infinitesimal increments $d\in D$, while the classical formula has a remainder controlled as the real increment $h$ tends to zero.

At second order, consider the sumset\footnote{Here $D+D$ denotes the pointwise sum in $R$. We never consider categorical sums in this paper.}
\[
D+D=\{d_0+d_1\mid d_0,d_1\in D\}.
\]
Lemma \ref{taylor_two_aux} gives, for every function $f\colon R\to R$, every $x\in R$ and every $d_0,d_1\in D$, the exact identity
\[
f(x+d_0+d_1)=f(x)+\partial f(x)(d_0+d_1)+\partial^2 f(x)d_0d_1.
\]
For comparison, if $g\in C^3(U,\mathbb{R})$ and $x\in U$ is fixed, the case $n=1$, $l=2$ of Theorem \ref{TaylorThm} gives
\[
g(x+h)=g(x)+g'(x)h+g''(x)\frac{h^2}{2}+O(|h|^3)\qquad\text{as }h\to0.
\]

Lemma \ref{taylor_two_aux} generalizes to Proposition \ref{taylor_two} and Proposition \ref{taylor_k_aux}, upon replacing $D+D$ with $D_2$, respectively $\underbrace{D+\dots +D}_\text{$k$ times}$. More generally, according to Theorem \ref{taylor_k}, $f$ coincides with its order $k$ Taylor polynomial on a $D_k$-neighborhood. Compare that with the case $n=1, l=k$ of Theorem \ref{TaylorThm}. Let $(k_0,\dots,k_{n-1}) \in \N^n$. We finally come to Theorem \ref{taylor_multi}, which states that a function of $n$ variables equals its Taylor polynomial of total degree at most $k_0+\cdots+k_{n-1}$ on the infinitesimal $(D_{k_0} \times \dots \times D_{k_{n-1}})$-neighborhood. The main calculatory steps in the proof are carried out in Theorem \ref{taylor_k} and Proposition \ref{taylor_multi_aux}.

\section{Logical framework and \Lean implementation} \label{implementation}
\subsection{The axiom of unique choice}
In this section we describe the logical framework of this work and explain how we implemented it in practice using the \Lean proof assistant and its mathematical library \mathlib. Our main object of study is a class of rings satisfying the so-called Kock--Lawvere property (see Definition~\ref{IsKockLawvere} below). In standard mathematics, using the law of excluded middle, it is easy to prove that such rings do not exist (see Subsection~\ref{subsection: excluded middle} for details) and in particular our work is incompatible with classical reasoning.

In \cite{SDG}, the author uses the internal language of a topos as a foundation for SDG. This work, on the other hand, uses the logical framework of the \Lean proof assistant: it is a version of dependent type theory with inductive types, a non-cumulative hierarchy of universes, including a impredicative universe of propositions and proof irrelevance. The interested reader can consult \cite{mario} for more details (note that \cite{mario} is about \Lean3 and we use here \Lean4, but the theory is essentially the same). \Lean's standard library adds three axioms to this framework:
\begin{itemize}
  \item \leanlink{https://github.com/riccardobrasca/lean4/blob/2191f60d235aff349ff9321ae7c016c0afe4527d/src/Init/Core.lean#L1603-L1612}{propext}: two equivalent propositions are equal.
  \item \leanlink{https://github.com/riccardobrasca/lean4/blob/2191f60d235aff349ff9321ae7c016c0afe4527d/src/Init/Prelude.lean#L526-L543}{Quot.sound}: enables to form quotients of types by equivalence relations.
  \item \leanlink{https://github.com/riccardobrasca/lean4/blob/2191f60d235aff349ff9321ae7c016c0afe4527d/src/Init/Prelude.lean#L843-L865}{Classical.choice}: the axiom of choice.
\end{itemize}
The first two axioms are compatible with constructive logic. The axiom of choice, on the other hand, is incompatible with constructivist reasoning: indeed in \Lean, the principle of excluded middle is \href{https://github.com/riccardobrasca/lean4/blob/2191f60d235aff349ff9321ae7c016c0afe4527d/src/Init/Classical.lean#L35-L67}{proved} using the axiom of choice, following an argument due to Diaconescu (see the first corollary in \cite{diaconescu})\footnote{Diaconescu proved that a weak form of the axiom of choice is equivalent to the law of excluded middle, see \cite{diaconescu}.}. In particular, we need here to avoid any use of the axiom of choice.

The Kock--Lawvere property for a ring $R$ states that, given a function $g \colon D \to R$, there is a unique element $b \in R$ such that, for all $d \in D$, a certain property $P(b, d)$ holds. The element $b$ depends of course on $g$, and the first step in the theory is to view this dependency as a function from $R^R$ to $R$ (the actual construction is slightly more involved since it starts with an element $x \in R$ and it considers the function $d \mapsto f(x + d)$, but this does not matter here). The main point is that we need the following:
\begin{principle}
Let $P$ be any property depending on two variables $x \in X$ and $y \in Y$. Suppose that for all $x \in X$ there is a unique $y \in Y$ such that $P(x, y)$ holds. We obtain a function $f \colon X \to Y$ such that for all $x \in X$ we have that $P(x, f(x))$ holds (and additionally $f(x)$ is the unique element of $Y$ with this property).
\end{principle}
In set theory functions are not a primitive notion and are defined by their graph. In the setting above, using the axiom of specification, one immediately gets a subset $S \subseteq X \times Y$ by considering the pairs $(x, y) \in X \times Y$ such that $P(x,y)$: this subset is the function $f$. In particular, the above principle trivially holds in set theory. In \Lean's type theory, functions are a primitive notion and they can be constructed essentially only via lambda abstraction. Since there is no way to extract a witness from an existential statement (more precisely, large elimination of a proposition is incompatible with proof irrelevance), the above construction is not doable without an additional axiom. Building $f$ using \Lean's standard library is of course very easy, using the axiom of choice (that allows to extract a witness from a nonempty type). However, as we explained above, we cannot use it. Instead, we use the following \href{https://github.com/riccardobrasca/SDG/blob/e4cd622492aba9be421ec0abcf152f58a0ee98df/SDG/Axiom/UniqueChoice.lean#L51-L53}{weaker axiom}, called \emph{the axiom of unique choice}: it allows to extract a witness from statements of unique existence (in the following $\alpha$ and $\beta$ are arbitrary types).
\begin{lstlisting}
axiom axiom_unique_choice (h : ∃! (_ : α), True) : α
\end{lstlisting}
We then proceed to prove various immediate corollaries that in particular allow one to select an element having a given property provided that such an element exists and is unique. At the end we obtain the \href{https://github.com/riccardobrasca/SDG/blob/e4cd622492aba9be421ec0abcf152f58a0ee98df/SDG/Axiom/UniqueChoice.lean#L74-L91}{following definition and results} that correspond to the principle stated above and are the only consequences of the axiom of unique choice that we will use in our development.
\begin{lstlisting}
/-- Given a property `P : α → β → Prop` such that for all `a : α`, there exists a unique `b : β` with `P a b`, then `unique_choice_fun h` gives a function `α → β` selecting this unique `b` for each `a`. -/
noncomputable def unique_choice_fun {P : α → β → Prop}
    (h : ∀ a, ∃! b, P a b) : α → β := ...

lemma unique_choice_fun_spec {P : α → β → Prop}
    (h : ∀ a, ∃! b, P a b) (a : α) : P a (unique_choice_fun h a) := ...

lemma unique_choice_fun_unique {P : α → β → Prop}
  (h : ∀ a, ∃! b, P a b) {a : α} {b : β} (hb : P a b) :
    unique_choice_fun h a = b := ...
\end{lstlisting}

\subsection{\Lean implementation}
We make extensive use of \Lean's standard library and the mathematical library \mathlib (see \cite{mathlib}). In both libraries very little effort is made to limit the use of the axiom of choice to proofs where it is really needed and this makes our work more delicate.

The \Lean command \lean{#print axioms decl} prints the axioms the declaration \lean{decl} transitively depends on, making it possible to check that we never use the axiom of choice in our theory. Checking by hand every declaration is impractical, so we use a \href{https://github.com/riccardobrasca/SDG/blob/e4cd622492aba9be421ec0abcf152f58a0ee98df/SDG/Linters/choice.lean}{linter}, written essentially by Damiano Testa, that emits a warning when a declaration depends on the axiom of choice, similarly to the built-in linter that flags uses of the \leanlink{https://github.com/riccardobrasca/lean4/blob/2191f60d235aff349ff9321ae7c016c0afe4527d/src/Init/Prelude.lean#L778-L794}{sorryAx} axiom.

The need to avoid the axiom of choice makes some proofs more involved than usual, especially since advanced tactics freely use it and it is difficult to control their behavior (in practice the only thing one can do is to split the proof into smaller steps until the tactic does not use the axiom of choice anymore). For example \lean{grind} (a tactic inspired by modern SMT solvers) always proves statements by contradiction and in particular every proof made by it uses the axiom of choice. In practice this means that \lean{grind} is unusable in our context. More importantly, this means that every declaration that transitively depends on a result proved using \lean{grind} depends on the axiom of choice, and therefore cannot be used in our development. Since \lean{grind} is used a lot in the library, this implies that the axiom of choice can be a dependency of many (sometimes totally unexpected) declarations. For example, we discovered that the mere fact of writing \lean{(2 : R)}, where $R$ is a ring, used the axiom of choice. The reason was that in order for \Lean to write \lean{(2 : R)}, an instance of the typeclass \lean{OfNat R 2} is needed, which is found via instance search (\leanlink{https://github.com/riccardobrasca/mathlib4/blob/c2934be473385c5b1e3a639cad5c00547d65e354/Mathlib/Data/Nat/Cast/Defs.lean#L38-L45}{instOfNatAtLeastTwo}). Checking its type, one sees that this declaration requires \lean{Nat.AtLeastTwo 2}, that again is found via \href{https://github.com/riccardobrasca/mathlib4/blob/c2934be473385c5b1e3a639cad5c00547d65e354/Mathlib/Data/Nat/Init.lean#L421-L424}{type class resolution}. This is the (mathematically trivial) fact that, given a non-zero natural number $n$, one has $n + 1 \ge 2$ and it was proved using the \lean{lia} tactic (a weaker version of \lean{grind}) and in particular depended on the axiom of choice. This has since been \href{https://github.com/leanprover-community/mathlib4/pull/32467}{fixed} in \mathlib by Anatole Dedecker.

On the other hand, tactics like \lean{ring} (used to perform algebraic manipulations) or \lean{zify} (a tactic that allows to cast statements about natural numbers to integers) are safer to use, so we do not have to write extremely verbose proofs, especially when dealing with algebraic computations. Still, the \lean{ring} tactic unpredictably used the axiom of choice when dealing with numeric literals.

As it will become apparent below, we perform quite a lot of manipulations with finite sums and products, so we make extensive use of this part of \mathlib. One realizes very quickly that almost all results about sums use the axiom of choice, and it is not always easy to understand the reason. Indeed \mathlib is extremely intricate, for example, we may need a declaration \lean{X} that depends on choice, so the first idea is just to provide \lean{X'}, the same statement but with a choice-free proof. The problem with this approach is that usually the proof of \lean{X} is very short, most likely a trivial application of \lean{Y} and \lean{Z}, which also depend on choice and whose proofs are again trivial applications of \lean{Y'} and \lean{Z'} and so on. In practice it is infeasible to find by hand the place where the axiom of choice is really used, and even if one finds it, it is then necessary to rewrite the entire chain of proofs that depend on it (such a chain can be very long even for mathematically easy results). For example, we discovered that the instance \leanlink{https://github.com/riccardobrasca/mathlib4/blob/c2934be473385c5b1e3a639cad5c00547d65e354/Mathlib/Data/Fintype/Basic.lean#L32-L33}{Fintype (Fin n)} (the fact that the set $\{0, \ldots, n - 1\}$ is finite) depends on various trivial arithmetic results, some of them proved using \lean{grind} and in particular it uses the axiom of choice. Providing a new instance by hand (together with various related results) is possible, but it requires to rewrite dozens of declarations, and this approach cannot scale.

The solution we adopted is to use a custom version of \mathlib that depends less on the axiom of choice. To explain how this works in practice, let us consider the declaration \leanlink{https://github.com/leanprover-community/mathlib4/blob/72e4c7171434a6437ee233ca79dd4db7a4703f72/Mathlib/Data/Multiset/UnionInter.lean#L157-L160}{Multiset.union_add_distrib}, distributivity of union over addition for multisets. In \mathlib (at the time of writing), the proof is as follows:
\begin{lstlisting}
lemma union_add_distrib (s t u : Multiset α) :
    s ∪ t + u = s + u ∪ (t + u) := by
  simpa [(· ∪ ·), union, eq_comm, Multiset.add_assoc, Multiset.add_left_inj] using
    show s + u - (t + u) = s - t by
      rw [t.add_comm, Multiset.sub_add_eq_sub_sub, Multiset.add_sub_cancel_right]
\end{lstlisting}
One then realizes that, say, \leanlink{https://github.com/riccardobrasca/mathlib4/blob/c2934be473385c5b1e3a639cad5c00547d65e354/Mathlib/Data/Multiset/AddSub.lean#L334}{Multiset.sub_add_eq_sub_sub} depends on choice, and its proof is:
\begin{lstlisting}
lemma sub_add_eq_sub_sub : s - (t + u) = s - t - u := by
    ext; simp [Nat.sub_add_eq]
\end{lstlisting}
Its dependence on choice is hidden in the use of \lean{simp} (a tactic that uses various results in the library to automatically simplify statements and hypotheses). This quickly gets out of hand. To better visualize the problem, we wrote a \href{https://github.com/riccardobrasca/SDG/blob/e4cd622492aba9be421ec0abcf152f58a0ee98df/scripts/graph.lean}{script}, with the assistance of Claude Sonnet 4.6, that creates the graph of dependencies on choice for a given declaration. The graph starts with the given declaration and ends with the axiom of choice, erasing all transitive dependencies that do not involve the axiom of choice. In the case of \leanlink{https://github.com/riccardobrasca/mathlib4/blob/c2934be473385c5b1e3a639cad5c00547d65e354/Mathlib/Data/Multiset/UnionInter.lean#L157-L158}{Multiset.union_add_distrib}, the graph is shown below.
\begin{figure}[p]
  \centering
  \makebox[\textwidth][c]{\includegraphics[height=0.9\textheight,keepaspectratio]{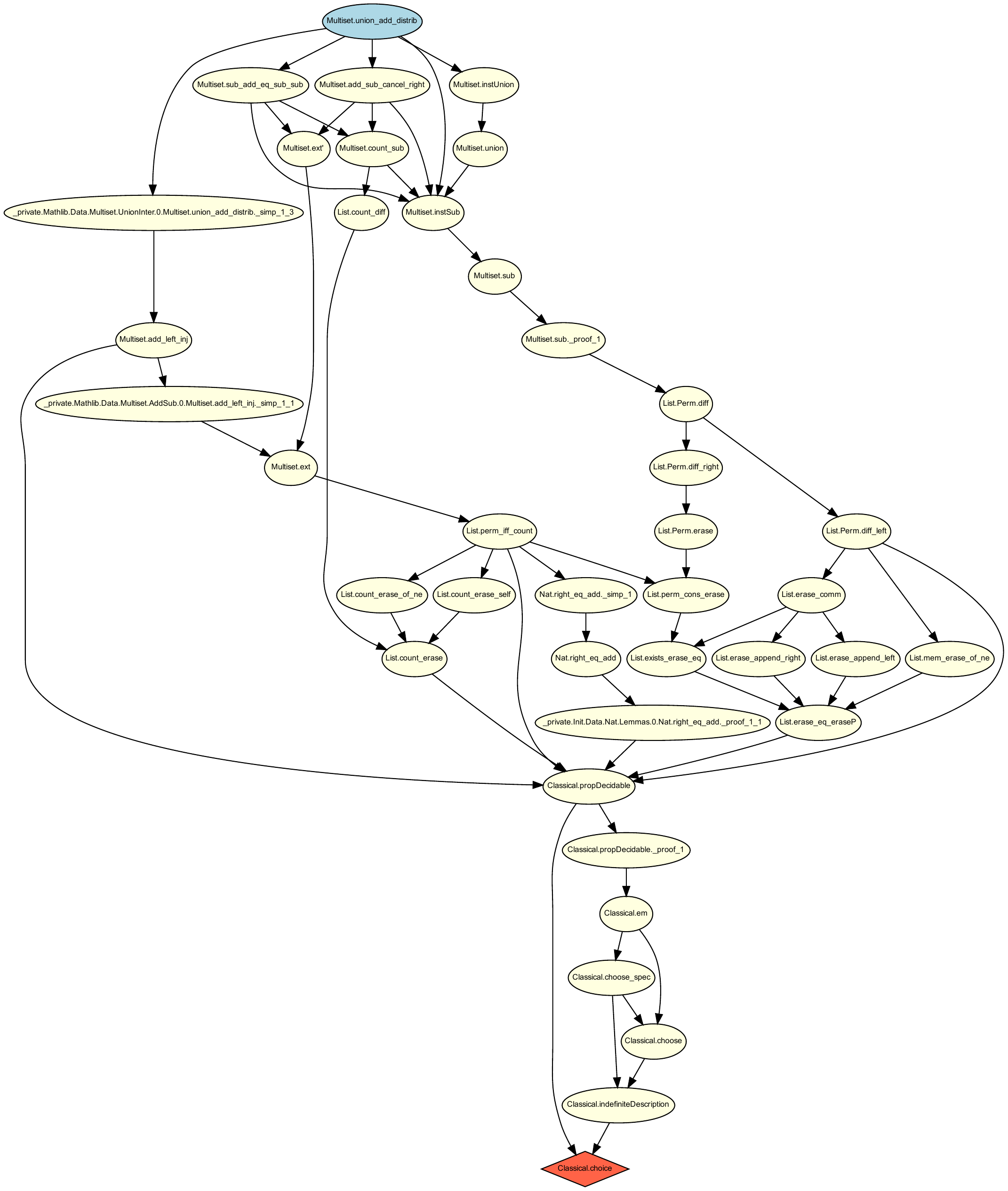}}
  \par\medskip
  {\small The dependency graph of \texttt{Multiset.union\_add\_distrib}.}
\end{figure}
One can clearly identify a zone (at the bottom of the graph) of declarations that do not hold in intuitionistic logic and should be avoided. In practice one can either disconnect a node from \leanlink{https://github.com/riccardobrasca/mathlib4/blob/c2934be473385c5b1e3a639cad5c00547d65e354/Mathlib/Data/Multiset/UnionInter.lean#L157-L158}{Multiset.union_add_distrib} (making its proof avoid that node) or remove the node from the graph (modifying its own proof to avoid choice). The first strategy is sometimes easier but the second one is more effective, as the modification will propagate more through the library. In the specific case of \leanlink{https://github.com/riccardobrasca/mathlib4/blob/c2934be473385c5b1e3a639cad5c00547d65e354/Mathlib/Data/Multiset/UnionInter.lean#L157-L158}{Multiset.union_add_distrib}, we give a new proof that
\begin{itemize}
    \item avoids the use of \leanlink{https://github.com/riccardobrasca/mathlib4/blob/c2934be473385c5b1e3a639cad5c00547d65e354/Mathlib/Data/Multiset/AddSub.lean#L127}{Multiset.add_left_inj} (whose dependence on choice is difficult to remove; see the discussion \hyperref[choice_sum_le_sum]{below});
    \item still uses, like the original proof, \leanlink{https://github.com/riccardobrasca/lean4/blob/2191f60d235aff349ff9321ae7c016c0afe4527d/src/Init/Data/List/Count.lean#L410-L428}{List.count_erase}: this declaration is in the core library and itself depended on choice, but we rewrote its proof to avoid it.
\end{itemize}
With these changes, our proof of \lean{Multiset.union_add_distrib} is choice-free.

To give a more complex example, we discovered that the instance \leanlink{https://github.com/riccardobrasca/mathlib4/blob/c2934be473385c5b1e3a639cad5c00547d65e354/Mathlib/Algebra/Field/Rat.lean#L27-L36}{Rat.instField}, the fact that rational numbers form a field, also depends on the axiom of choice. Identifying the source of this dependency by hand is extremely difficult: for instance, even the bare definition of addition on $\Q$ turns out to depend on choice. The graph for \leanlink{https://github.com/riccardobrasca/mathlib4/blob/c2934be473385c5b1e3a639cad5c00547d65e354/Mathlib/Algebra/Field/Rat.lean#L27-L36}{Rat.instField} can be found \href{https://github.com/riccardobrasca/SDG/blob/e4cd622492aba9be421ec0abcf152f58a0ee98df/scripts/Rat.instField.svg}{here}. Although it is very intricate, it is still easy to identify a small subset of nodes whose removal suffices to make this instance choice-free. In this particular case, it was enough to change the proof of just seven declarations: six trivial arithmetic lemmas in core (for example \leanlink{https://github.com/riccardobrasca/lean4/blob/2191f60d235aff349ff9321ae7c016c0afe4527d/src/Init/Data/Nat/Lemmas.lean#L267-L268}{Nat.le_iff_lt_add_one}), and \leanlink{https://github.com/riccardobrasca/mathlib4/blob/c2934be473385c5b1e3a639cad5c00547d65e354/Mathlib/Data/Rat/Defs.lean#L139-L140}{Rat.zero_ne_one} in \mathlib. None of these declarations posed any particular difficulty. Interestingly, in several cases it was enough to split a double implication proved by \lean{omega} (a tactic that automatically solves goals about natural numbers and integers) into its two directions, both again proved by \lean{omega}, to eliminate the dependence on choice. This somewhat ad hoc approach was sufficient in practice for our project.

The \mathlib version we use can be found \href{https://github.com/leanprover-community/mathlib4/pull/35685/}{here}, and it can be used in the standard way by modifying the \texttt{lakefile.toml} and \texttt{lean-toolchain} files (as done in our \href{https://github.com/riccardobrasca/SDG/blob/e4cd622492aba9be421ec0abcf152f58a0ee98df/lakefile.toml}{project}), except that getting the cache of compiled files is slightly more involved than usual (see the instructions \href{https://github.com/riccardobrasca/SDG/blob/e4cd622492aba9be421ec0abcf152f58a0ee98df/README.md#getting-the-cache}{here}). Note that since certain declarations we needed to tweak are in \Lean's core library and in \texttt{batteries} (one of the \Lean basic \href{https://github.com/leanprover-community/batteries}{libraries}), we also use custom versions of both (see \href{https://github.com/riccardobrasca/batteries/tree/v4.35.0-less-choice}{here} and \href{https://github.com/riccardobrasca/lean4/tree/v4.35.0-less-choice}{here}). In the end, this yielded a version of \mathlib in which more than sixteen thousand declarations that previously depended on the axiom of choice are now choice-free, particularly in the area of finite sets and rational numbers.

We now describe various interesting issues we encountered while trying to avoid the axiom of choice.
\begin{enumerate}
    \item As explained above we had to rewrite proofs using \lean{grind} or \lean{lia}. This made certain proofs quite verbose but no serious problem arose. Similarly, certain proofs started with \lean{classical} (that allows the automatic use of choice-dependent declarations) just to keep them short.
    \item We found proofs written by contradiction that could be easily rewritten in a direct way.
    \item We discovered that commutativity of addition of multisets used the axiom of choice. The reason is that the \leanlink{https://github.com/riccardobrasca/mathlib4/blob/c2934be473385c5b1e3a639cad5c00547d65e354/Mathlib/Algebra/Order/Group/Multiset.lean#L37-L42}{AddCommMonoid} instance on \lean{Multiset α}
    was defined directly inside the \href{https://github.com/leanprover-community/mathlib4/blob/72e4c7171434a6437ee233ca79dd4db7a4703f72/Mathlib/Algebra/Order/Group/Multiset.lean#L37-L43}{AddCancelCommMonoid} instance, and cancellativity of addition uses choice (in a non-fundamental but harder to remove way). Our solution here was simply to separate the two instances without modifying any proof.
    \item Proofs regarding families of multisets were more interesting. For example, declaration \leanlink{https://github.com/riccardobrasca/mathlib4/blob/c2934be473385c5b1e3a639cad5c00547d65e354/Mathlib/Data/Multiset/Bind.lean#L422-L431}{Multiset.Nodup.sigma} says that the dependent sum of a family of multisets\footnote{See the definition of \href{https://github.com/riccardobrasca/mathlib4/blob/c2934be473385c5b1e3a639cad5c00547d65e354/Mathlib/Data/Multiset/Bind.lean\#L378-L381}{\texttt{Multiset.sigma}} for details about this notion.} without duplicates over a multiset without duplicates is again without duplicates. The idea of the proof is basically to go through lists as multisets are by definition equivalence classes of lists up to permutations. For a single multiset there is no need to choose a representative (that would require choice), and one can argue by induction on the quotient. However, this statement involves a family of multisets, i.e. a term of type \lean{(a : α) → Multiset (σ a)} (here $\sigma$ parametrizes the family) and it is not possible to lift such a dependent function to \lean{(a : α) → List (σ a)} without choice. Such proofs had to be genuinely rewritten, something that was not always easy.
    \item\label{choice_sum_le_sum} An interesting case is \leanlink{https://github.com/riccardobrasca/mathlib4/blob/c2934be473385c5b1e3a639cad5c00547d65e354/Mathlib/Algebra/Order/BigOperators/Group/Finset.lean#L114-L119}{Finset.sum_le_sum} and related declarations (for example \lean{Multiset.add_left_inj}): the proof in \mathlib uses the axiom of choice to obtain an instance of decidable equality on the type that indexes the family, and it seems difficult to avoid this. The set of indices we consider in this work is a finite product of natural numbers, so it already has decidable equality, and it is not difficult to provide a choice-free alternative to \leanlink{https://github.com/riccardobrasca/mathlib4/blob/c2934be473385c5b1e3a639cad5c00547d65e354/Mathlib/Algebra/Order/BigOperators/Group/Finset.lean#L114-L119}{Finset.sum_le_sum} that works in our setting. On the other hand, it is not clear how to integrate this into \mathlib, since our version is technically weaker, having an additional assumption. Modifying \mathlib would require assuming decidable equality in all theorems and definitions depending on \leanlink{https://github.com/riccardobrasca/mathlib4/blob/c2934be473385c5b1e3a639cad5c00547d65e354/Mathlib/Algebra/Order/BigOperators/Group/Finset.lean#L114-L119}{Finset.sum_le_sum}, thereby changing many declarations. Our solution, although not completely satisfactory, is to provide in our project \leanlink{https://github.com/riccardobrasca/SDG/blob/e4cd622492aba9be421ec0abcf152f58a0ee98df/SDG/Axiom/BigOperators.lean#L26-L34}{SDG.Finset.sum_le_sum}, a version of \leanlink{https://github.com/riccardobrasca/mathlib4/blob/c2934be473385c5b1e3a639cad5c00547d65e354/Mathlib/Algebra/Order/BigOperators/Group/Finset.lean#L114-L119}{Finset.sum_le_sum} that does not use choice, and ignore the declaration in \mathlib.
    \item Some of our proofs use well-founded recursion, so \Lean generates termination goals and proves them automatically with a default tactic. In some cases (for example \leanlink{https://github.com/riccardobrasca/SDG/blob/e4cd622492aba9be421ec0abcf152f58a0ee98df/SDG/IsKockLawvere/Taylor.lean#L120-L156}{SDG.taylor_k_aux'}) the proofs found in this way depended on the axiom of choice, and we had to provide the termination proofs manually in a \lean{decreasing_by} clause.
    \item In our version of \mathlib, the instance \lean{Field ℚ} (\leanlink{https://github.com/riccardobrasca/mathlib4/blob/c2934be473385c5b1e3a639cad5c00547d65e354/Mathlib/Algebra/Field/Rat.lean#L27-L36}{Rat.instField}) is choice-free. However, to apply basic results such as \leanlink{https://github.com/riccardobrasca/mathlib4/blob/c2934be473385c5b1e3a639cad5c00547d65e354/Mathlib/Algebra/Group/DivInvMonoid.lean#L326-L328}{mul_inv_rev}, which states that $(ab)^{-1}=b^{-1}a^{-1}$, \Lean needs to find an instance of \lean{DivisionMonoid ℚ}. This is an auxiliary algebraic structure expressing laws for multiplication and inversion. Type class resolution obtains this instance through the mathematically trivial declaration \leanlink{https://github.com/riccardobrasca/mathlib4/blob/c2934be473385c5b1e3a639cad5c00547d65e354/Mathlib/Algebra/GroupWithZero/Basic.lean#L370-L382}{GroupWithZero.toDivisionMonoid}, whose proof uses classical decidability to distinguish zero from nonzero elements. Consequently, applying \lean{mul_inv_rev} to rational numbers introduces a dependency on the axiom of choice, even though equality on $\Q$ is decidable without choice. Knowing that the particular arguments are nonzero does not remove this dependency, since it comes from the generic instance. We therefore add \href{https://github.com/riccardobrasca/SDG/blob/e4cd622492aba9be421ec0abcf152f58a0ee98df/SDG/Axiom/Instances.lean}{higher-priority, choice-free instances} so that type class resolution can obtain \lean{DivisionMonoid ℚ} without introducing a dependency on choice.
    \item We managed to avoid changing any definition in \mathlib and to avoid using the axiom of unique choice there, keeping our version as close as possible to the standard library.
\end{enumerate}
During the writing of this paper, Bhavik Mehta and Jovan Gerbscheid (after a discussion with Mario Carneiro) suggested a different approach that is more maintainable in the long term: instead of modifying existing declarations in \mathlib, one can reprove them and use a custom command, similar to \lean{#print axioms}, that, during the checking of the axioms used, instructs \Lean to ignore the original declaration and consider the new one instead. This way, we can keep the original library intact and maintain a separate file with choice-free versions of the necessary declarations. This also allows one to keep the original library up to date and to easily identify the declarations that we need to reprove. We then discovered that this approach was already implemented by Shuhao Song, together with various tools to tweak definitions; see \href{https://github.com/znssong/Frucht/blob/458c314dcd53131a72d4e1c048cf7853fa0bbfef/Frucht/Classical.lean}{here} for details.

\subsection{Excluded middle} \label{subsection: excluded middle} We conclude this section with a discussion about the law of excluded middle and its incompatibility with synthetic differential geometry.

The law of excluded middle (LEM) asserts that for any proposition $p$, either $p$ or $\neg p$ is true. In particular, one can reason by cases, and define functions by specifying the image of any $x$ in the case where $p(x)$ holds and in the case where $\neg p(x)$ holds (here $p$ is a property of $x$).
The following lemma (\leanlink{https://github.com/riccardobrasca/SDG/blob/e4cd622492aba9be421ec0abcf152f58a0ee98df/SDG/IsKockLawvere_one/EM.lean#L28-L38}{SDG.false_of_isKockLawvere_one}) is the only place in our development where we allow the axiom of choice, through the tactic \lean{classical}. Accordingly, the linter described in Section~\ref{implementation} is disabled in the corresponding file.
\begin{lemma}[Exercise 1.1 of \cite{SDG}]
The existence of a $1$-Kock--Lawvere ring (see Definition~\ref{IsKockLawvere}) contradicts excluded middle (and so the axiom of choice).
\end{lemma}
\begin{proof}
Let $R$ be a $1$-Kock--Lawvere ring. First, $\neg(\forall d \in D, d = 0)$ holds without using LEM (see \leanlink{https://github.com/riccardobrasca/SDG/blob/e4cd622492aba9be421ec0abcf152f58a0ee98df/SDG/IsKockLawvere_one/Basic.lean#L44-L47}{SDG.D_ne_zero}). Indeed, if every element of $D$ were $0$, then both $b = 0$ and $b = 1$ would satisfy $g(d) = g(0) + bd$ for all $d \in D$, where $g$ is the zero function, and the uniqueness in Definition~\ref{IsKockLawvere} would give $0 = 1$. Using LEM, we obtain $d_0 \in D$ with $d_0 \neq 0$. Using LEM again, we can define the function $g \colon D \to R$ by
\[
g(d)=\begin{cases}0 \mbox{ if }d=0 \\
1 \mbox{ if }d\neq 0\end{cases}
\]
By Definition~\ref{IsKockLawvere}, there is $b \in R$ such that $g(d) = g(0) + bd$ for all $d \in D$. Evaluating at $d_0$ gives
\[
1=g(d_0)=g(0)+b d_0=b d_0
\]
Therefore
\[
1 = 1^2 = (b d_0)^2 = b^2 d_0^2 = 0
\]
since $d_0^2 = 0$, which is impossible as $0\neq 1$ in $R$.
\end{proof}

\section{Basics}\label{basics}
\subsection{Basic definitions and results}

Let $R$ be a commutative ring. We define the $k^{th}$ order \href{https://github.com/riccardobrasca/SDG/blob/e4cd622492aba9be421ec0abcf152f58a0ee98df/SDG/Basic/Defs.lean#L30-L33}{infinitesimal object} as follows.
\begin{definition}
 \label{SDG.D}
For any natural number $k$, we denote by $D_k$ the subsemigroup of $R$
\[
\{r \in R \mbox{ such that } r ^ {k+1} = 0 \}.
\]

We use the convention (\leanlink{https://github.com/riccardobrasca/SDG/blob/e4cd622492aba9be421ec0abcf152f58a0ee98df/SDG/Basic/Defs.lean#L35-L36}{SDG.D})
\[
D \colonequals D_1=\{r \in R \mbox{ such that } r ^ 2 = 0 \}
\]
\end{definition}
The following lemmas are obvious (see \href{https://github.com/riccardobrasca/SDG/blob/e4cd622492aba9be421ec0abcf152f58a0ee98df/SDG/Basic/Defs.lean#L49-L50}{here} and \href{https://github.com/riccardobrasca/SDG/blob/e4cd622492aba9be421ec0abcf152f58a0ee98df/SDG/Basic/D.lean#L30-L48}{here}).
\begin{lemma}
  \label{SDG.zero_mem_D}
  We have that $0 \in D_k$ for all $k$.
\end{lemma}

\begin{lemma} \label{closure_prod}
Let $x,y \in R$. If either $x$ or $y$ belongs to $D_k$, then the product of $x$ and $y$ belongs to $D_k$.
\end{lemma}

\begin{lemma}
    \label{containDk}
    For all $k \leq \ell$, we have that $D_k \subseteq D_{\ell}$.
\end{lemma}
\begin{lemma} \label{square}
Let $x, y \in D$. Then $(x+y)^2=2xy$.
\end{lemma}
The following lemma (\leanlink{https://github.com/riccardobrasca/SDG/blob/e4cd622492aba9be421ec0abcf152f58a0ee98df/SDG/Basic/D.lean#L57-L65}{SDG.mem_D_add_pow}), even if completely elementary, is the key technical point that will render the proofs of most results in Section~\ref{1varsection} more direct than those in \cite{SDG}.
\begin{lemma}\label{synthle}
For any $k$, if $x \in D$ and $y \in R$, we have
\[
(x+y)^{k+1}= (k+1)x y^k+y^{k+1}
\]
\end{lemma}
\begin{proof}
This is an induction argument. The claim is trivial when $k=0$. Assuming that the claim is true for a given $k$, we have that
\begin{gather*}
        (x+y)^{k+2}=\\
        (x+y)((k+1)x y^k+y^{k+1})=\\
        (k+1)x^2y^k+xy^{k+1}+(k+1)xy^{k+1}+y^{k+2} =\\
        (k+2)xy^{k+1}+y^{k+2}
\end{gather*} since $x^2=0$.
\end{proof}
The following (\leanlink{https://github.com/riccardobrasca/SDG/blob/e4cd622492aba9be421ec0abcf152f58a0ee98df/SDG/Basic/D.lean#L67-L73}{SDG.mem_D_sum_pow_succ}) will also be useful in the next section.
\begin{lemma}\label{synthzero}
For any $k$ and $b_0,\dots,b_{k-1} \in D$, we have
\[
\left( \sum_{i < k} b_i \right)^{k+1}=0
\]
\end{lemma}
\begin{proof}
We prove the lemma by induction. If $k=0$ the sum is empty and there is nothing to prove. Assume that the lemma holds for $k$. Then, using Lemma \ref{synthle}, we have that
\begin{gather*}
        \left(\sum_{i < k+1} b_i\right)^{k+2}=\left(b_0 + \sum_{i < k} b_{i+1}\right)^{k+2}\\
        =(k+2)b_0\left(\sum_{i < k} b_{i+1}\right)^{k+1} + \left(\sum_{i < k} b_{i+1}\right)^{k+2} = 0
\end{gather*}
\end{proof}
Finally, the following identity (\leanlink{https://github.com/riccardobrasca/SDG/blob/e4cd622492aba9be421ec0abcf152f58a0ee98df/SDG/Basic/D.lean#L75-L81}{SDG.mem_D_sum_pow}) will be used in the proof of Theorem \ref{taylor_k}.
\begin{lemma} \label{sum_pow_eq_mul_prod}
For any $k$ and $b_0,\dots,b_{k-1} \in D$, we have
\[
\left( \sum_{i < k} b_i \right)^k=k!\prod_{i < k} b_i
\]
\end{lemma}
\begin{proof}
We prove the lemma by induction. If $k=0$ the sum and the product are empty and there is nothing to prove (we use the convention that $0^0 = 1$). Assume that the lemma holds for $k$. Then, using Lemmas \ref{synthle} and \ref{synthzero}, we have that
\begin{gather*}
        \left(\sum_{i < k+1} b_i\right)^{k+1}=\left(b_0 + \sum_{i < k} b_{i+1}\right)^{k+1} =\\
        (k+1)b_0\left(\sum_{i < k} b_{i+1}\right)^k + \left(\sum_{i < k} b_{i+1}\right)^{k+1} = \\
        (k+1)b_0k!\prod_{i < k} b_{i+1} = (k+1)!\prod_{i < k+1} b_i
\end{gather*}
\end{proof}

\begin{definition} \label{IsKockLawvere}
We say that $R$ is $1$-Kock--Lawvere (\leanlink{https://github.com/riccardobrasca/SDG/blob/e4cd622492aba9be421ec0abcf152f58a0ee98df/SDG/Basic/Defs.lean#L62-L65}{SDG.IsKockLawvere_one}) if it is nontrivial and for all $g \colon D \to R$ there exists a unique $b \in R$ such that $g(d) = g(0) + b d$ for all $d \in D$. We say that $R$ is Kock--Lawvere (\leanlink{https://github.com/riccardobrasca/SDG/blob/e4cd622492aba9be421ec0abcf152f58a0ee98df/SDG/Basic/Defs.lean#L67-L71}{SDG.IsKockLawvere}) if it is nontrivial and if for all $k$ and all $g \colon D_k \to R$ there exist unique $b_0,\dots, b_{k-1} \in R$ such that, for all $d \in D_k$, we have
\[
g(d) = g(0) + \sum_{i < k} b_i d^{i+1}
\]
Obviously, being Kock--Lawvere \href{https://github.com/riccardobrasca/SDG/blob/e4cd622492aba9be421ec0abcf152f58a0ee98df/SDG/Basic/Defs.lean#L73-L81}{implies being $1$-Kock--Lawvere}.
\end{definition}
\begin{remark}\label{alpha}
The $1$-Kock--Lawvere condition admits the following categorical formulation, which is the one used in the introduction (see \cite[I.1]{SDG}). Consider the map
\[
\alpha \colon R \times R \to R^D, \qquad (a, b) \mapsto (d \mapsto a + bd)
\]
Then $R$ is $1$-Kock--Lawvere if and only if it is nontrivial and $\alpha$ is bijective. Indeed, since $\alpha(a, b)(0) = a$, the surjectivity of $\alpha$ amounts to the existence of $b$ in Definition \ref{IsKockLawvere}, and its injectivity to the uniqueness. Under the identification of $R \times R$ with the dual numbers ${R[x]}/{\langle x^2 \rangle}$, $\alpha$ is the evaluation map of the Kock--Lawvere axiom schema of the introduction for $W = {R[x]}/{\langle x^2 \rangle}$. The equivalence is formalized in our repository (\leanlink{https://github.com/riccardobrasca/SDG/blob/e4cd622492aba9be421ec0abcf152f58a0ee98df/SDG/Basic/Alpha.lean#L32-L44}{SDG.isKockLawvere_one_iff_bijective_alpha}).
\end{remark}

\section{One variable differential calculus}\label{1varsection}\label{1var}
\subsection{\texorpdfstring{$1$-Kock--Lawvere rings}{1-Kock--Lawvere rings}}
We assume in this subsection that $R$ is $1$-Kock--Lawvere.
\begin{definition} \label{deriv}
Let $f \colon R \to R$ be any function. We define $\partial f \colon R \to R$, the derivative of $f$ (\leanlink{https://github.com/riccardobrasca/SDG/blob/e4cd622492aba9be421ec0abcf152f58a0ee98df/SDG/IsKockLawvere_one/Deriv.lean#L45-L74}{SDG.deriv}), as follows. Let $x \in R$ and consider the function $f_x \colon D \to R$ that sends $d$ to $f(x + d)$. By assumption, there is a unique element $b_x \in R$ such that for all $d \in D$ we have
\[
f_x(d) = f_x(0) + b_x d
\]
Using the axiom of unique choice, this gives a function $\partial f \colon R \to R$ that maps $x$ to $b_x$.
\end{definition}
\begin{remark}
One could instead include the uniquely determined coefficients as data in the definition of a Kock--Lawvere ring, thereby defining derivatives without a separate use of unique choice. We follow the standard formulation in \cite[I.1]{SDG}, which expresses the Kock--Lawvere axiom as a unique-existence property. In the internal logic of a topos the two formulations are equivalent, since unique choice is available there \cite[Proposition II.3.4]{SDG}. In \Lean's type theory the data formulation would be a stronger assumption, as explained in Section~\ref{implementation}, and the axiom of unique choice is exactly what is needed to pass from the unique-existence formulation to the function $\partial f$.
\end{remark}
\begin{remark}\label{higherderiv}
In the rest of the paper we will often consider higher order derivatives, denoted by $\partial^n f$. These are defined in the obvious recursive way, where $\partial^0 f = f$. Note that in \Lean the notation for the $n^{th}$ order derivative is \href{https://github.com/riccardobrasca/SDG/blob/e4cd622492aba9be421ec0abcf152f58a0ee98df/SDG/IsKockLawvere_one/Deriv.lean#L32-L33}{\texttt{\ensuremath{\partial}\textasciicircum[n]f}}.
\end{remark}
The following (\leanlink{https://github.com/riccardobrasca/SDG/blob/e4cd622492aba9be421ec0abcf152f58a0ee98df/SDG/IsKockLawvere_one/Basic.lean#L26-L30}{SDG.cancel_d}) is an important remark that will play a crucial role in the proofs of the results in the whole paper.
\begin{proposition}  \label{universaldcancels}
Let $b_1, b_2 \in R$. If for all $d\in D$ one has $b_1 d=b_2 d$, then $b_1=b_2$.
\end{proposition}
\begin{proof}
Consider the functions $g_1,g_2 \colon D\to R$ defined by $g_1(d)=b_1 d$ and $g_2(d)=b_2 d$. By assumption we have $g_1=g_2$ and so from the uniqueness part of Definition \ref{IsKockLawvere} it follows that $b_1=b_2$.
\end{proof}
The following theorem (\leanlink{https://github.com/riccardobrasca/SDG/blob/e4cd622492aba9be421ec0abcf152f58a0ee98df/SDG/IsKockLawvere_one/Deriv.lean#L38-L39}{SDG.taylor_one}) says that in the synthetic setting all functions are affine linear on a $D$-neighborhood.
\begin{theorem}[Theorem 2.1 of \cite{SDG}]\label{taylor_one}
Let $f\colon R \to R$ be a function. For all $x \in R$ and $d\in D$, we have
\[
f(x+d)=f(x)+\partial f(x) d
\]
Moreover, $\partial f(x)$ is characterized by this property (\leanlink{https://github.com/riccardobrasca/SDG/blob/e4cd622492aba9be421ec0abcf152f58a0ee98df/SDG/IsKockLawvere_one/Deriv.lean#L41-L43}{SDG.derivative_unique}).
\end{theorem}
\begin{proof}
Let $x \in R$ be arbitrary. By Definition~\ref{deriv}, the derivative $\partial f$ satisfies
\[
f_x(d) = f_x(0) + \partial f (x) d
\]
Since $f_x(d) = f(x + d)$, the stated formula follows. For the last assertion, if $r \in R$ satisfies $f(x+d) = f(x) + rd$ for all $d \in D$, then $rd = \partial f(x) d$ for all $d \in D$, and Proposition \ref{universaldcancels} gives $r = \partial f(x)$.
\end{proof}
We have that $\partial$ is an $R$-linear operator (\leanlink{https://github.com/riccardobrasca/SDG/blob/e4cd622492aba9be421ec0abcf152f58a0ee98df/SDG/IsKockLawvere_one/Deriv.lean#L45-L74}{SDG.deriv}) and moreover it satisfies the Leibniz rule (\leanlink{https://github.com/riccardobrasca/SDG/blob/e4cd622492aba9be421ec0abcf152f58a0ee98df/SDG/IsKockLawvere_one/Deriv.lean#L80-L82}{SDG.deriv_mul}): for all functions $f$ and $g$, we have
\[
\partial(f \cdot g) = \partial f \cdot g + f \cdot \partial g.
\]
Indeed, this is a straightforward application of Theorem \ref{taylor_one}. Let $x\in R$ and $d\in D$. We then have that
 \begin{gather*}
         (fg)(x)+\partial (fg)(x)d=(fg)(x+d)=f(x+d)g(x+d)=\\
         (f(x)+\partial f(x)d)(g(x)+\partial g(x)d)=\\
         f(x)g(x)+(\partial f(x) g(x)+f(x)\partial g(x))d,
 \end{gather*}
so
\[
\partial (fg)(x)d=(\partial f(x) g(x)+f(x)\partial g(x))d
\]
which, by Proposition \ref{universaldcancels}, reduces to the Leibniz rule.

\begin{lemma}
    \label{diffproperties}
We have the \href{https://github.com/riccardobrasca/SDG/blob/e4cd622492aba9be421ec0abcf152f58a0ee98df/SDG/IsKockLawvere_one/Deriv.lean#L84-L124}{following} basic differentiation formulas.
\begin{enumerate}
    \item \label{deriv_const} Let $r\in R$ and $c\colon R\to R$ be the constant function $c(x)=r$ for all $x\in R$. Then $\partial c=0$.
    \item \label{deriv_id} Let $\id \colon R \to R$ be the identity function. Then $\partial \id=1$.
    \item (Chain rule) Let $f,g\colon R\to R$ be functions. Then, the derivative of the composition $f\circ g$ is given as $\partial (f\circ g)(x)=\partial f (g(x)) \partial g(x)$ for all $x\in R$.
    \item (Power rule) Let $n\in \N$ and $f \colon R \to R$ be the function $x \mapsto x^n$. Then $\partial f(x)=n x^{n-1}$. We will simply write in the sequel $\partial(x^n) = n x^{n-1}$.
    \item If $f\colon R\to R$ is a function such that $\frac{1}{f}$ is defined everywhere, then $\partial \left(\frac{1}{f}\right)=-\frac{\partial f}{f^2}$.
    \item If $f$ is invertible then for all $x\in R$, we have that $(\partial f) \circ f^{-1}(x) \in R^\ast$ and
\[
\partial f^{-1}=\frac{1}{(\partial f) \circ f^{-1}}
\]
\end{enumerate}
\end{lemma}
\begin{proof}
Let $x\in R$ and $d\in D$.
\begin{enumerate}
    \item By Definition \ref{deriv}, we have that $r=r+\partial c (x) d$, so $0=\partial c (x) d$, and so $\partial c (x)=0$ by Proposition \ref{universaldcancels}.
    \item By Definition \ref{deriv}, $x+d=x+\partial \id (x) d$, so $d=\partial \id (x) d$. Using Proposition \ref{universaldcancels} we obtain $\partial \id (x)=1$.
    \item Applying Definition \ref{deriv} to $g$ and using the fact that $g(x)\in R$ and $\partial g(x) d \in D$ (see Lemma \ref{closure_prod}) to expand $f(g(x)+\partial g(x) d)$, we have that
    \begin{gather*}
         (f\circ g)(x+d)=f(g(x+d)) =\\
         f(g(x)+\partial g(x) d) = \\
         (f\circ g)(x)+\partial f (g(x)) \partial g(x) d
    \end{gather*}
    Comparing this with Definition \ref{deriv} applied to $f\circ g$ and cancelling $(f\circ g)(x)$, we obtain
    \[
    \partial(f\circ g)(x)d = \bigl(\partial f(g(x))\partial g(x)\bigr)d.
    \]
    Since this holds for every $d\in D$, Proposition \ref{universaldcancels} gives the chain rule.
    \item We can prove this by induction. If $n=0$ we can apply \eqref{deriv_const}. Assume that the claim is true for $k$. Then, from the Leibniz rule, the inductive hypothesis and \eqref{deriv_id}, it follows that for any $x\in R$,
    \begin{gather*}
        \partial (x^{k+1})=\partial (x^k x) =\\
        (\partial x^k)x+x^k (\partial x) =\\
        (kx^{k-1})x+x^k=(k+1)x^k
    \end{gather*}
    and we are done.
    \item Since
    \[
    \partial \left(\frac{f}{f}\right)=\partial 1=0
    \]
    by \eqref{deriv_const}, and since by the Leibniz rule
    \[
    \partial \left(\frac{f}{f}\right)=\frac{\partial f}{f}+f\partial \left(\frac{1}{f}\right)
    \]
    it follows that
    \[
    \partial \left(\frac{1}{f}\right)=-\frac{\partial f}{f^2}
    \]
    \item Observe that by \eqref{deriv_id} and the chain rule,

    \begin{gather*}
         1=\partial (f\circ f^{-1}) =\\
         ((\partial f) \circ f^{-1})\partial f^{-1}
    \end{gather*}
    so
    \[
    \partial f^{-1}=\frac{1}{(\partial f) \circ f^{-1}}
    \]
\end{enumerate}
\end{proof}
Applying the first-order Taylor formula twice gives the following identity (\leanlink{https://github.com/riccardobrasca/SDG/blob/e4cd622492aba9be421ec0abcf152f58a0ee98df/SDG/IsKockLawvere_one/Deriv.lean#L128-L132}{SDG.taylor_two_add}).
\begin{lemma}[Proposition 2.3 of \cite{SDG}]\label{taylor_two_aux}
Let $d_0, d_1 \in D$. For any function $f\colon R \to R$ and for all $x \in R$, we have
\[
f(x+d_0+d_1)=f(x)+\partial f(x)(d_0+d_1)+\partial^2 f(x)d_0d_1.
\]
\end{lemma}
\begin{proof}
Applying Theorem \ref{taylor_one} to $f$ and $\partial f$, we obtain
\begin{gather*}
    f(x+d_0+d_1)=f(x+d_0)+\partial f(x+d_0) d_1 =\\
    f(x)+\partial f (x) d_0+(\partial f(x)+\partial^2f(x) d_0) d_1=\\
    f(x)+\partial f(x) (d_0+d_1)+ \partial^2 f(x) d_0 d_1.
\end{gather*}
\end{proof}
\subsection{Kock--Lawvere rings}
We additionally suppose that $R$ is Kock--Lawvere and that $2$ is invertible in $R$. The following (\leanlink{https://github.com/riccardobrasca/SDG/blob/e4cd622492aba9be421ec0abcf152f58a0ee98df/SDG/IsKockLawvere/Taylor.lean#L25-L44}{SDG.taylor_two}) generalizes Lemma \ref{taylor_two_aux} above to $D_2$-neighborhoods. It is technically not needed in the development of the theory, as it is a special case of Theorem \ref{taylor_k}, but we include it here for completeness and to show how the proof can be simplified in this case.
\begin{proposition}[Theorem 3.1 of \cite{SDG}]\label{taylor_two}
Let $f\colon R \to R$ be any function. For all $x \in R$ and $\delta\in D_2$, we have
\[
f(x+\delta)=f(x)+\partial f(x) \delta+\partial^2 f(x) \frac{\delta^2}{2}
\]
\end{proposition}
\begin{proof}
Let us define a function $g_x \colon D_2 \to R$ by $g_x(d)=f(x+d)$. By Definition \ref{IsKockLawvere}, there are (unique) $b_0$ and $b_1$ in $R$ such that, for all $\delta \in D_2$, we have $g_x(\delta)=g_x(0)+b_0 \delta+b_1 \delta^2$, i.e.
\[
    f(x+\delta)=f(x)+b_0 \delta+b_1 \delta^2
\]
In particular, if $d \in D \subseteq D_2$, we have $f(x+d)=f(x)+b_0 d$, so $b_0 = \partial f(x)$ by the characterization of the derivative in Theorem \ref{taylor_one}.

It remains to prove that $b_1 = \frac{\partial^2 f(x)}{2}$. To that end, let $d_0$ and $d_1$ be in $D$, so $d_0 + d_1 \in D_2$. By Lemma \ref{square},
\[
f(x + d_0 + d_1) = f(x) + \partial f(x)(d_0+d_1)+2b_1d_0d_1
\]
but also Lemma \ref{taylor_two_aux} gives
\[
f(x + d_0 + d_1) = f(x) + \partial f(x)(d_0+d_1)+\partial^2 f(x) d_0d_1
\]
Comparing these identities and applying Proposition \ref{universaldcancels} twice gives $2b_1=\partial^2 f(x)$. Since $2$ is invertible, we deduce the desired form of $b_1$ and we are done.
\end{proof}
We now generalize the two previous results to any $D_k$. We assume that $R$ is a $\Q$-algebra. The following (\leanlink{https://github.com/riccardobrasca/SDG/blob/e4cd622492aba9be421ec0abcf152f58a0ee98df/SDG/IsKockLawvere/Taylor.lean#L68-L108}{SDG.taylor_k_aux}) also generalizes Lemma \ref{taylor_two_aux} but to $\underbrace{(D+\dots +D)}_\text{$k$ times}$-neighborhoods.
As above, $D+\dots+D$ denotes the set of sums of $k$ elements of $D$ in $R$.
\begin{proposition}\label{taylor_k_aux}
Let $d_0,\dots,d_{k-1} \in D$ and $\delta=d_0+\dots+d_{k-1}$. For any function $f\colon R \to R$ and for all $x \in R$ we have
\[
f(x+\delta)= \sum_{n < k + 1} \partial^n f(x) \frac{\delta^n}{n!} =f(x)+\partial f(x) \delta+\dots+\partial^k f(x) \frac{\delta^k}{k!}
\]
\end{proposition}
\begin{proof}
The proof is by induction. If $k=0$ the result is obvious (the cases $k=1,2$ also follow from Theorem \ref{taylor_one} and Lemmas \ref{taylor_two_aux} and \ref{square}).

Assume that the formula holds for any function and any $k$-tuple and let $d_0,\dots,d_k \in D$. We set
\[
\delta \colonequals \sum_{n < k} d_{n+1} \mbox{ and } \Delta \colonequals \sum_{n < k + 1} d_n
\]
so $\Delta = d_0 + \delta$. After an elementary manipulation, we need to prove
\begin{gather}
f(x+\Delta)= \sum_{n < k + 2} \partial^n f(x) \frac{\Delta^n}{n!} = \nonumber\\
f(x)+ \sum_{n < k} \partial^{n+1} f(x) \frac{\Delta^{n+1}}{(n+1)!} +\partial^{k+1} f(x) \frac{\Delta^{k+1}}{(k+1)!} \label{reformulation}
\end{gather}
By Theorem \ref{taylor_one} and the induction hypothesis (applied to $f$ and $\partial f$) we have
\begin{gather*}
f(x + \Delta) = f((x + \delta) + d_0) =\\
f(x + \delta) + \partial f(x + \delta)d_0 = \\
\sum_{n < k + 1} \partial^n f(x) \frac{\delta^n}{n!} + \sum_{n < k + 1} \partial^{n+1} f(x) \frac{\delta^n}{n!}d_0 = \\
f(x)+\sum_{n < k} \partial^{n+1} f(x) \frac{\delta^{n+1}}{(n+1)!} + \sum_{n < k} \partial^{n+1} f(x) \frac{\delta^n}{n!}d_0 + \partial^{k+1} f(x) \frac{\delta^k}{k!}d_0=\\
f(x)+ \sum_{n < k} \partial^{n+1}f(x)\frac{\delta^{n+1}+(n+1)\delta^nd_0}{(n+1)!}+\partial^{k+1} f(x) \frac{\delta^k}{k!}d_0
\end{gather*}
We compare the last line with \eqref{reformulation} term by term. Using Lemma \ref{synthle}, we have, for $0\leq n\leq k$,
\[
\Delta^{n+1}=(d_0 + \delta)^{n+1}=(n+1)d_0\delta^n+\delta^{n+1}
\]
For $n<k$, this identifies the corresponding summands of the two sums. For the remaining term $n=k$, Lemma \ref{synthzero} gives $\delta^{k+1}=0$, since $\delta$ is a sum of $k$ elements of $D$. Thus
\[
\frac{\Delta^{k+1}}{(k+1)!}=\frac{\delta^k}{k!}d_0
\]
which identifies the final terms and completes the proof.
\end{proof}
We now move on to the main result of this section (\leanlink{https://github.com/riccardobrasca/SDG/blob/e4cd622492aba9be421ec0abcf152f58a0ee98df/SDG/IsKockLawvere/Taylor.lean#L158-L172}{SDG.taylor_k}), which generalizes the previous one to any $D_k$-neighborhood.
\begin{theorem}[Theorem 3.1 of \cite{SDG}]\label{taylor_k}
Let $f\colon R \to R$ be any function. For all $x \in R$ and $\delta\in D_k$ we have
\[
f(x+\delta)= \sum_{n < k + 1}\partial^nf(x) \frac{\delta^n}{n!} = f(x)+\partial f(x) \delta+\dots+\partial^k f(x)  \frac{\delta^k}{k!}
\]
\end{theorem}
\begin{proof}
If $k=0$ the result is obvious (note that we also proved the cases $k=1,2$ in Theorem \ref{taylor_one} and Proposition \ref{taylor_two}), so we fix an arbitrary $k$ and prove the statement for $D_{k+1}$. Note that, unlike that of Proposition \ref{taylor_k_aux}, the proof is not by induction on $k$: the statement for $D_{k+1}$ is deduced directly from Proposition \ref{taylor_k_aux}.

We consider the function $g_x \colon D_{k+1} \to R$ defined by $g_x(d)=f(x+d)$: by Definition \ref{IsKockLawvere}, there exist (unique) $b_0,\dots,b_k \in R$ such that for all $\Delta\in D_{k+1}$ we have
\begin{gather}\label{equk}
        f(x+\Delta)=f(x)+ \sum_{i < k + 1} b_i \Delta^{i+1}
\end{gather}
Take such a $\Delta$. It is enough to prove that, for all $n < k+1$, we have
\[
b_n=\frac{\partial^{n+1}f(x)}{(n+1)!}
\]
We prove this claim itself by strong induction on $n$.
\begin{itemize}
    \item To prove that $b_0=\partial f(x)$ it suffices, by the characterization of the derivative in Theorem \ref{taylor_one}, to show that $f(x + d) = f(x) + b_0d$ for all $d \in D$.\footnote{Here, as in the proof of Proposition \ref{taylor_two}, the first coefficient is identified through the characterization of the derivative in Theorem \ref{taylor_one}, and the higher ones through Proposition \ref{universaldcancels}. These are two forms of the same fact: as shown in the proof of Theorem \ref{taylor_one}, the characterization is Proposition \ref{universaldcancels} applied to the two expansions of $f(x+d)$. The formal proofs follow the same pattern, using \href{https://github.com/riccardobrasca/SDG/blob/e4cd622492aba9be421ec0abcf152f58a0ee98df/SDG/IsKockLawvere_one/Deriv.lean\#L41-L43}{\texttt{SDG.derivative\_unique}} for $b_0$ and \href{https://github.com/riccardobrasca/SDG/blob/e4cd622492aba9be421ec0abcf152f58a0ee98df/SDG/IsKockLawvere_one/Basic.lean\#L26-L30}{\texttt{SDG.cancel\_d}}, or its iterated version, for the other coefficients.} This follows immediately from \eqref{equk}, since $D \subseteq D_{k+1}$ and $d^2 = 0$ for $d \in D$.
    \item Let us prove the claim for $n + 1$ (assuming that $n+1 < k + 1$, and that the claim holds for all $m \leq n$). Since $(n+2)!$ is invertible in $R$, we have to show that $(n+2)!\,b_{n+1}=\partial^{n+2}f(x)$. By Proposition \ref{universaldcancels}, it suffices to show that $(n+2)!\,b_{n+1}d_0=\partial^{n+2}f(x)d_0$ for all $d_0 \in D$. Applying Proposition \ref{universaldcancels} to the two elements $(n+2)!\,b_{n+1}d_0$ and $\partial^{n+2}f(x)d_0$, it suffices in turn to show that $(n+2)!\,b_{n+1}d_0d_1=\partial^{n+2}f(x)d_0d_1$ for all $d_0, d_1 \in D$. Repeating this argument, after applying Proposition \ref{universaldcancels} $n+2$ times in total, it is enough to prove that, for all $d_0, \ldots, d_{n+1} \in D$, we have
    \[
    (n+2)!b_{n+1}\prod_{i < n + 2}d_i=\partial^{n+2}f(x)\prod_{i < n + 2}d_i
    \]
    The point of this reduction is that, by Lemma \ref{sum_pow_eq_mul_prod}, the product $\prod_{i < n + 2}d_i$ is, up to the factor $(n+2)!$, the power $\delta^{n+2}$ of the sum $\delta$ of the $d_i$, and \eqref{equk} gives information on $f(x+\delta)$ in terms of the powers of $\delta$.
    We set $\delta \colonequals \sum_{i < n + 2}d_i$. By Lemma \ref{synthzero}, we have that $\delta \in D_{n+2} \subseteq D_{k+1}$ and in particular
\[
f(x+\delta)=f(x)+ \sum_{i < k + 1} b_i \delta^{i+1}
\]
Using Proposition \ref{taylor_k_aux} to replace $f(x+\delta)$ we get, after an obvious re-indexing,
\[
\sum_{i < n + 2}\partial^{i+1}f(x) \frac{\delta^{i+1}}{(i+1)!} = \sum_{i < k + 1} b_i \delta^{i+1}
\]
We now use that $b_0 = \partial f(x)$ and re-indexing again leads to
\[
\sum_{i < n + 1}\partial^{i+2}f(x) \frac{\delta^{i+2}}{(i+2)!} = \sum_{i < k} b_{i+1} \delta^{i+2}
\]
Since $\delta \in D_{n+2}$, if $i \geq n+1$ we have $\delta^{i+2} = 0$, so
\[
\sum_{i < n + 1}\partial^{i+2}f(x) \frac{\delta^{i+2}}{(i+2)!} = \sum_{i < n + 1} b_{i+1} \delta^{i+2}
\]
hence
\[
\sum_{i < n}\partial^{i+2}f(x) \frac{\delta^{i+2}}{(i+2)!} + \partial^{n+2}f(x)\frac{\delta^{n+2}}{(n+2)!} = \sum_{i < n} b_{i+1} \delta^{i+2}+b_{n+1} \delta^{n+2}
\]
If $i < n$, we can apply strong induction to $i+1$, so
\[
b_{i+1}=\frac{\partial^{i+2}f(x)}{(i+2)!}
\]
and in particular the two sums in the last formula are equal, so
\[
\partial^{n+2}f(x)\frac{\delta^{n+2}}{(n+2)!} = b_{n+1} \delta^{n+2}
\]
Using Lemma \ref{sum_pow_eq_mul_prod} (remember that $\delta = \sum_{i < n + 2}d_i$), this equality becomes
\[
\partial^{n+2}f(x)\prod_{i < n + 2}d_i = b_{n+1} (n+2)! \prod_{i < n + 2}d_i
\]
and that is exactly what we needed to prove.
\end{itemize}
\end{proof}

\section{Several variables differential calculus}\label{multi-var}

\subsection{Partial derivatives}
We assume in this subsection that $R$ is $1$-Kock--Lawvere.
\begin{definition}
\label{derivpartial}
Let $f \colon R^n \to R$ be any function. For each $i<n$, we define $\partial_i f \colon R^n \to R$, the $i^{th}$ partial derivative of $f$ (\leanlink{https://github.com/riccardobrasca/SDG/blob/e4cd622492aba9be421ec0abcf152f58a0ee98df/SDG/IsKockLawvere_one/PartialDeriv.lean#L62-L81}{SDG.partial_deriv}), as follows. Let $x \in R^n$ and consider the function $f^i_x \colon D \to R$ that sends $d$ to $f(x_0,\dots,x_i + d,\dots,x_{n-1})$. By assumption, there is a unique element $b^i_x \in R$ such that for all $d \in D$ we have
\[
f^i_x(d) = f^i_x(0) + b^i_x d
\]
Using the axiom of unique choice, this gives a function $\partial_i f \colon R^n \to R$ that maps $x$ to $b^i_x$.
\end{definition}
\begin{remark} \label{higherderivpartial}
Similarly to Remark \ref{higherderiv}, we denote the higher partial derivatives of $f$ by $\partial_i^n f$ (in \Lean those are denoted by \href{https://github.com/riccardobrasca/SDG/blob/e4cd622492aba9be421ec0abcf152f58a0ee98df/SDG/IsKockLawvere_one/PartialDeriv.lean#L48-L49}{\texttt{\ensuremath{\partial}\_[i]\textasciicircum[n]f}}).
\end{remark}
The following (\leanlink{https://github.com/riccardobrasca/SDG/blob/e4cd622492aba9be421ec0abcf152f58a0ee98df/SDG/IsKockLawvere_one/PartialDeriv.lean#L51-L55}{SDG.partial_taylor_one}) is an immediate consequence of the definition of partial derivative and it generalizes Theorem \ref{taylor_one}.
\begin{theorem}\label{taylor_one_multi}
Let $f\colon R^n \to R$ be a function. For all $x \in R^n$ and $d\in D$, we have
\[
f(x_0,\dots,x_i + d,\dots,x_{n-1})=f(x)+\partial_i f (x) d
\]
Moreover, $\partial_i f$ is characterized by this property (\leanlink{https://github.com/riccardobrasca/SDG/blob/e4cd622492aba9be421ec0abcf152f58a0ee98df/SDG/IsKockLawvere_one/PartialDeriv.lean#L56-L60}{SDG.partial_derivative_unique}).
\end{theorem}
\begin{proof}
Let $x \in R^n$ be arbitrary. By Definition~\ref{derivpartial}, the partial derivative $\partial_i f(x)$ satisfies
\[
f^i_x(d) = f^i_x(0) + \partial_i f(x) d
\]
Since $f^i_x(d) = f(x_0,\dots,x_i + d,\dots,x_{n-1})$ the theorem is obvious.
\end{proof}
One has various trivial \href{https://github.com/riccardobrasca/SDG/blob/e4cd622492aba9be421ec0abcf152f58a0ee98df/SDG/IsKockLawvere_one/PartialDeriv.lean#L90-L169}{results} about partial derivatives, for example linking the derivative of a function of one variable $f \colon R \to R$ with its (unique) partial derivative. The following (\leanlink{https://github.com/riccardobrasca/SDG/blob/e4cd622492aba9be421ec0abcf152f58a0ee98df/SDG/IsKockLawvere_one/PartialDeriv.lean#L171-L184}{SDG.partial_deriv_comm}) is a less trivial, important property of partial derivatives.
\begin{proposition}\label{partials_commute}
For any function $f\colon R^n \to R$,
\[
\partial_i \partial_j f=\partial_j \partial_i f
\]
\end{proposition}
\begin{proof}
Let $x \in R^n$, and $d_i,d_j \in D$. We can suppose $i < j$. On the one hand we have that
\begin{gather*}
f(x_0,\dots,x_i+d_i,\dots,x_j+d_j,\dots,x_{n-1})-f(x)=\\
f(x_0,\dots,x_i+d_i,\dots,x_j+d_j,\dots,x_{n-1})-\\
f(x_0,\dots,x_i+d_i,\dots,x_j,\dots,x_{n-1})+\\
f(x_0,\dots,x_i+d_i,\dots,x_j,\dots,x_{n-1})-f(x)=\\
\partial_j f (x_0,\dots,x_i+d_i,\dots,x_j,\dots,x_{n-1})d_j+\partial_i f(x)d_i =\\
\partial_j f(x)d_j+\partial_i \partial_j f(x) d_i d_j+\partial_i f(x)d_i
\end{gather*}
But on the other hand we have that
\begin{gather*}
f(x_0,\dots,x_i+d_i,\dots,x_j+d_j,\dots,x_{n-1})-f(x)=\\
f(x_0,\dots,x_i+d_i,\dots,x_j+d_j,\dots,x_{n-1})-\\
f(x_0,\dots,x_i,\dots,x_j + d_j,\dots,x_{n-1})+\\
f(x_0,\dots,x_i,\dots,x_j + d_j,\dots,x_{n-1})-f(x)=\\
\partial_i f(x)d_i+\partial_j \partial_i f(x) d_j d_i+\partial_j f(x)d_j
\end{gather*}
From these computations we deduce that
\[
\partial_i \partial_j f(x) d_i d_j=\partial_j \partial_i f(x) d_i d_j
\]
Hence after canceling the universally quantified $d_j$, followed by the $d_i$, it follows that
\[
\partial_i \partial_j f(x)=\partial_j \partial_i f(x)
\]
\end{proof}
\subsection{Taylor theorems in several variables}
From now on we assume that $R$ is a $\Q$-algebra and that it is Kock--Lawvere.

The following result is not stated explicitly in \cite{SDG}, and it is not part of our formalization. Also, we have not used this theorem in any result of the paper. However, we thought it would be instructive to keep it as it guided our intuition towards proving all other multivariate Taylor theorems. Theorem \ref{taylor_n_multi_aux} is a special case of Proposition \ref{taylor_multi_aux} (see the paragraph preceding that proposition).

For any $f \colon R^n \to R$, any $x \in R^n$ and any decidable\footnote{In constructive mathematics there are several inequivalent notions of finiteness. In the absence of the law of excluded middle, which is incompatible with our axioms by Subsection~\ref{subsection: excluded middle}, an arbitrary subset of a finite set need not be finite in any of these senses. A subset $H \subseteq \{0,\dots,n-1\}$ is decidable if, for every $i$, either $i \in H$ or $i \notin H$. Decidable subsets correspond to elements of $\{0,1\}^n$, so there are exactly $2^n$ of them, and they are exactly the subsets a classical mathematician has in mind.} $H \subseteq \{0,\dots,n-1\}$, we write
\[
x^H=\prod_{j\in H} x_j \mbox{ and } \partial_H f = \partial_{j_1} \cdots \partial_{j_m} f, \mbox{ where } H = \{j_1, \dots, j_m\}.
\]

Thus $\partial_H$ is the composition of the partial derivatives $\partial_j$ for $j \in H$. Note that by Proposition \ref{partials_commute}, the order in which they are composed is immaterial. For $H = \emptyset$ we set $\partial_{\emptyset} f = f$. All subsets $H$ appearing in the rest of this subsection are understood to be decidable, and a sum indexed by $H \subseteq \{0,\dots,n-1\}$ ranges over the $2^n$ decidable subsets.
\begin{theorem}\label{taylor_n_multi_aux}
Let $f\colon R^n \to R$ be a function. For any $x\in R^n$ and $d\in D^n$, we have
\[
f(x+d)=\sum_{H \subseteq \{0,\dots,n-1\}} \partial_H f(x) d^H
\]
\end{theorem}
\begin{proof}
We induct on $n$. If $n=0$, this is obvious.

Assume that the result is true for $n$. Let $f\colon R^{n+1} \to R$ be any function, let $x\in R^{n+1}$ and $d\in D^{n+1}$. Then,
\begin{gather*}
f(x+d)=f(x+d)-f(x_0+d_0,\dots,x_{n-1}+d_{n-1},x_n)+\\
f(x_0+d_0,\dots,x_{n-1}+d_{n-1},x_n)=\\
f(x_0+d_0,\dots,x_{n-1}+d_{n-1},x_n)+\\
\partial_n f(x_0+d_0,\dots,x_{n-1}+d_{n-1},x_n) d_n
\end{gather*}
where we used Theorem \ref{taylor_one_multi} to get the second term in the sum. Now, observe that by the induction hypothesis, fixing the last variable, we have that
\[
f(x_0+d_0,\dots,x_{n-1}+d_{n-1},x_n)=\sum_{H \subseteq \{0,\dots,n-1\}} \partial_H f(x) d^H
\]
where to write $d^H$ we view a given $H \subseteq \{0,\dots,n-1\}$ as a subset of $\{0,\dots,n\}$ with the last component being $0$. Similarly, by the induction hypothesis, we have that
\[
\partial_n f(x_0+d_0,\dots,x_{n-1}+d_{n-1},x_n)=\sum_{H\subseteq \{0,\dots,n-1\}} \partial_H \partial_n f(x) d^H
\]
In particular
\begin{gather*}
f(x+d)=\sum_{H\subseteq \{ 0,\dots,n-1\}} \partial_H f(x) d^H +
\sum_{H\subseteq \{0,\dots,n-1\}} \partial_H \partial_n f(x) d^H d_n =\\
\sum_{H\subseteq \{0,\dots,n-1\}} \partial_H f(x) d^H + \sum_{H = H' \cup \{n\}, H' \subseteq \{0,\dots,n-1\}} \partial_H f(x) d^H
\end{gather*}
so it follows that
\begin{gather*}
f(x+d)= \sum_{H\subseteq \{0,\dots,n\}} \partial_H f(x) d^H
\end{gather*}
as required
\end{proof}
For any pair of multi-indices $\alpha, \beta \in \N^a$, let us write $\alpha \leq \beta$ if and only if $\alpha_i \leq \beta_i$ for all $i$. Moreover, define $|\alpha|=\sum_{j < a} \alpha_j$ and $\alpha!=\prod_{j < a} \alpha_j !$. Finally, for $x \in R^n$ with $a \leq n$, we set $x^{\alpha}=\prod_{j < a} x_j^{\alpha_j}$.

Let us consider now a function $f \colon R^n \to R$ and $\alpha \in \N^a$, where $a \leq n$. For any $i < a$, we have the higher partial derivative $\partial_i^{\alpha_i}$, and repeatedly applying those operators to $f$ we get the mixed derivative (\leanlink{https://github.com/riccardobrasca/SDG/blob/e4cd622492aba9be421ec0abcf152f58a0ee98df/SDG/IsKockLawvere/TaylorMulti.lean#L63-L65}{SDG.mixed_partial_deriv})
\[
\partial[\alpha] f
\]
\begin{remark}
As already explained in Remarks \ref{higherderiv} and \ref{higherderivpartial}, our notation differs from the usual analysis one, where the mixed derivative above would be denoted as
\[
\frac{\partial^{|\alpha|}f}{\partial x^{\alpha}}=\frac{\partial^{|\alpha|}f}{\partial x^{\alpha_0}_0 \dots \partial x^{\alpha_{a-1}}_{a-1}}
\]
It is impractical to use this notation in the formalization. Also, note that the additional flexibility of taking a multi-index $\alpha \in \N^a$ with $a \leq n$ instead of just $\alpha \in \N^n$ is usually ignored in informal mathematics, as one can just set the last components of $\alpha$ equal to $0$. Since various formal proofs will be by induction on $a$, this mathematically trivial generalization is important for us.
\end{remark}

The following (\leanlink{https://github.com/riccardobrasca/SDG/blob/e4cd622492aba9be421ec0abcf152f58a0ee98df/SDG/IsKockLawvere/TaylorMulti.lean#L198-L202}{SDG.taylor_multi_aux}) generalizes Theorem \ref{taylor_n_multi_aux} to the infinitesimal product neighborhood of $D_{k_i}$. Note that Theorem \ref{taylor_n_multi_aux} is the special case of Proposition \ref{taylor_multi_aux} when $k_i = 1$ for all $i$.
\begin{proposition}[Theorem 5.2 of \cite{SDG}]\label{taylor_multi_aux}
Let $k=(k_0,\dots,k_{n-1}) \in \N^n$, and $f\colon R^n \to R$ be any function. For all $x\in R^n$ and $d\in D_{k_0} \times \dots \times D_{k_{n-1}}$ we have
\[
f(x+d)=\sum_{\alpha \leq k} \partial[\alpha]f(x) \frac{d^{\alpha}}{\alpha!}
\]
\end{proposition}
\begin{proof}
We prove the proposition by induction on $n$. If $n = 0$ the statement is trivial.

The steps of the proof are as follows: firstly, we introduce a function $g$ of one variable from $f$, fixing the first $n$ variables and we apply Theorem \ref{taylor_k} to $g$. Secondly, we use the induction hypothesis on certain functions of $n$ variables obtained by fixing the last variable in various derivatives of $f$ (see the functions $h_i$ below). The final step is to use a reindexing function to obtain the desired form.

Let $f \colon R^{n+1} \to R$ be a function and let $k$, $x$ and $d$ be as in the statement. We set $\tilde k = (k_0, \ldots, k_{n-1})$ and similarly for $x$ and $d$. We introduce the function
\begin{gather*}
g \colon R \to R \\
y \mapsto f(\tilde x + \tilde d, y)
\end{gather*}
We have $f(x + d) = g(x_n + d_n)$, so by Theorem \ref{taylor_k} we get
\[
f(x + d) = g(x_n + d_n) = \sum_{i < k_n + 1}\partial^ig(x_n) \frac{d_n^i}{i!}
\]
Since by definition
\[
\partial^i g(x_n)=\partial_n^i f (\tilde x + \tilde d, x_n)
\]
we obtain
\[
f(x + d) = \sum_{i < k_n + 1} \partial_n^i f (\tilde x + \tilde d, x_n)\frac{d_n^i}{i!}
\]
For all $i$, we define a function
\begin{gather*}
h_i \colon R^n \to R \\
y \mapsto \partial_n^i f(y, x_n)
\end{gather*}
By definition we have
\[
h_i(\tilde x + \tilde d) = \partial_n^i f(\tilde x + \tilde d, x_n)
\]
so
\[
f(x + d) = \sum_{i < k_n + 1} h_i(\tilde x + \tilde d)\frac{d_n^i}{i!}
\]
and we can apply the induction hypothesis to each $h_i$, obtaining
\[
f(x + d) = \sum_{i < k_n + 1}\sum_{\tilde \alpha \leq \tilde k} \partial[\tilde \alpha]h_i(\tilde x) \frac{\tilde d^{\tilde \alpha}}{\tilde \alpha!}\frac{d_n^i}{i!}
\]
(here $\tilde \alpha$ is a dummy index ranging over $\N^n$). Since there is a bijection from the set $\{\alpha \leq k\}$ to $\{0, \ldots, k_n \} \times \{\tilde \alpha \leq \tilde k \}$ we get, after reindexing,
\[
f(x + d) = \sum_{\alpha \leq k} \partial[\tilde \alpha]h_{\alpha_n}(\tilde x) \frac{\tilde d^{\tilde \alpha}}{\tilde \alpha!}\frac{d_n^{\alpha_n}}{\alpha_n!}
\]
So it is enough to prove that for all $\alpha \leq k$ we have
\[
\partial[\tilde \alpha]h_{\alpha_n}(\tilde x) = \partial[\alpha]f(x)
\]
and this is true by definition of partial derivatives and the fact that they commute (see Proposition \ref{partials_commute}).
\end{proof}
The following (\leanlink{https://github.com/riccardobrasca/SDG/blob/e4cd622492aba9be421ec0abcf152f58a0ee98df/SDG/IsKockLawvere/TaylorMulti.lean#L227-L239}{SDG.taylor_multi}) is the main result of our paper and it is the synthetic analogue of Theorem \ref{TaylorThm}.
\begin{theorem}\label{taylor_multi}
Let $k=(k_0,\dots,k_{n-1}) \in \N^n$, and $f\colon R^n \to R$ be any function. For all $x\in R^n$ and $d\in D_{k_0} \times \dots \times D_{k_{n-1}}$, we have
\[
f(x+d)=\sum_{|\alpha| \leq |k|} \partial[\alpha]f(x) \frac{d^{\alpha}}{\alpha!}
\]
\end{theorem}
\begin{proof}
Let $S$ be the set of all $\alpha \in \N^n$ such that $|\alpha| \leq |k|$ with $\alpha_i > k_i$ for at least one $i$. Then, since for any $\alpha \in S$, we have that $d^{\alpha_i}_i=0$ for some $i$ and in particular $d^{\alpha}=0$, it follows that
\begin{gather*}
f (x + d) = \sum_{\alpha \leq k} \partial[\alpha]f(x) \frac{d^{\alpha}}{\alpha!} = \\
\sum_{\alpha \in S} \partial[\alpha]f(x) \frac{d^{\alpha}}{\alpha!}+\sum_{\alpha \leq k} \partial[\alpha]f(x) \frac{d^{\alpha}}{\alpha!}= \sum_{|\alpha| \leq |k|} \partial[\alpha]f(x) \frac{d^{\alpha}}{\alpha!}
\end{gather*}
where we used Proposition \ref{taylor_multi_aux} in the first line.
\end{proof}
All summands outside the finite decidable box $\{\alpha\in\N^n\mid\alpha\leq k\}$ vanish. The sum over all multi-indices in the following corollary denotes this finite sum.
\begin{corollary}\label{cor_taylor}
Let $f\colon R^n \to R$ be any function. For all $x\in R^n$ and $d\in D_{k_0} \times \dots \times D_{k_{n-1}}$, we have
\[
f(x+d)=\sum_{\alpha} \partial[\alpha]f(x) \frac{d^{\alpha}}{\alpha!}
\]
\end{corollary}

\bibliographystyle{amsalpha}
\bibliography{references}
\end{document}